\pdfoutput=1
\documentclass{article}
\usepackage[a4paper]{geometry}
\usepackage{framed}
\usepackage{xr}
\usepackage{bbm}
\usepackage{mathbbol}
\usepackage{amsmath}
\usepackage{amssymb}
\usepackage{gensymb}
\usepackage{listings}
\usepackage[arrow, matrix]{xy}
\usepackage{graphicx}
\usepackage[authoryear,round,longnamesfirst]{natbib}

\DeclareMathOperator*{\argmin}{\arg\!\min}

\newcommand{\Prob}{\mathbb{P}}

\newtheorem{thm}{Theorem}

\newtheorem{proof}{Proof}

\begin{document}

\title{On an Exact and Nonparametric Test for the Separability of Two Classes by Means of a Simple Threshold}

\author{FABIAN SCHROEDER\\
	Austrian Institute of Technology, Department of Molecular Diagnostics\\
	fabian.schroeder.fl@ait.ac.at	}

\maketitle

\begin{abstract}
This paper introduces a statistical test inferring whether a variable allows separating two classes by means of a single critical value. Its test statistic is the prediction error of a nonparametric threshold classifier. While this approach is adequate for univariate classification tasks, it is especially advantageous for filter-type variable selection. It constitutes a robust and nonparametric method which may identify important otherwise neglected variables. It can incorporate the operating conditions of the classification task. Last but not least, the exact finite sample distribution of the test statistic under the null hypothesis can be calculated using a fast recursive algorithm.  
\end{abstract}

\section{Introduction}
	
Consider the task of supervised classification with one variable and two classes. In this context, the discrete state or class of an instance is not directly observable, we can, however, observe one or many continuous variables which we expect to contain information about the state. Based on a random sample, in which the class memberships of the instances are known, we would like to derive a classifier and evaluate its success at predicting the class labels. We might also pose the question of the significance of the outcome and calculate its p-value, the probability that the outcome is at least as extreme as observed under the null hypothesis that the continuous random variable exhibits no information about the class. This can be characterized by means of the class-conditional distributions (CCDs) of the random variable. If the CCDs differ from one another, predicting the class membership based on the value of the continuous random variable is superior to random guessing.

The motivation for such a test arises from the need to preselect variables for a classification model on a univariate basis, often referred to as filtering, see, e.g., \citet{GuyonEtAl} for an introduction. This approach usually constitutes the first and crude step in a model selection procedure aimed at reducing the dimension of the sample space. 

Very popular for this purpose are location tests, which infer whether there is a significant difference in the central location of the CCDs. They, however, exhibit a number of shortcomings when used as variable filters, which we shall discuss as a motivation for the proposed test. Consider the t-test \citep{Student, Welch1947}, and their modifications, e.g. the moderated t-test \citep{smyth2004linear}, which are widely used as variable filters. These tests are based on the assumption of Gaussian CCDs. Even small deviations from these assumptions, e.g. a contamination with outliers or skewness, can strongly affect the result. They are not robust and only exact when the parametric assumptions are met. When applied to thousands of different variables this will most likely be problematic.

Nonparametric location tests are inexact in the sense that the exact finite sample distribution of the test statistic under the null hypothesis, also referred to as null distribution (ND), can rarely be obtained. A notable exception is the Wilcoxon Rank-Sum Test \citep{Wilcoxon45}, where exact inference is possible up to a sample size of about 30. Most methods, however, rely on asymptotic results, which is especially critical for small sample sizes and when the p-values of the test statistics used for filtering variables are corrected for multiple testing, see, e.g., \cite{BH}. Small errors in the calculation of the p-values can be amplified by the adjustment. 

The most common nonparametric approach is the permutation test, also known as randomization test \citep{Edgington2011}. It establishes the ND by calculating all possible values of the test statistic under rearrangements of the class labels. However, even for small sample sizes, an exhaustive calculation is computationally infeasible. For example, two different class labels with 30 samples each can be permuted in $ \binom{60}{30} > 10^{17}$ different ways. Thus, one resolves to draw permutations randomly. This approach is problematic because the p-values of the most promising variables will lie on the tail of the ND. Since only a few random permutations will yield values in this region of the ND, the estimation will be most inaccurate. \citet{fewerPerm} suggest to approximate the tail of the ND by curve fitting techniques. Thus, although permutation tests are sometimes referred to as exact tests, in practice exact p-values cannot be obtained. 

Furthermore, all location tests ignore the operating conditions (OCs) of the classification task. The OCs comprise the misclassification costs and the prevalence of the classes in the population, see \citet{CSOverview} for an overview. Think of medical diagnosis where it is the rule rather than the exception that the consequences of misdiagnosing an ill individual are more severe than those of misdiagnosing a healthy individual and that most diseases are relatively rare. Disregarding the OCs might lead to the selection of suboptimal variables.

In this paper, a statistical test is proposed that avoids these shortcomings. Its test statistic is based on the estimation of the prediction error of a classifier and, thus, naturally accounts for different operating conditions. Choosing the classifier is critical and must not be arbitrary. If the CCDs were known, we could simply choose the classifier that minimizes the prediction error. However, this is rarely the case and would furthermore require a different classifier for every variable. We will, thus, opt for a nonparametric approach assuming a simple type of functional form for the classifier, the threshold classifier. Due to its straightforward interpretability, it is very popular for medical applications, e.g. the use of a fever thermometer. If the measured temperature exceeds a certain critical value the individual is classified as ill. It can be motivated by parametric models, e.g. for two Gaussian class-conditional distributions with equal variances the Bayes classifier is of this type. However, since any assumption of a parametric form is problematic for the purpose of filtering, we will, thus, identify the optimal threshold classifier by minimizing the sample prediction error, also known as empirical risk, see \cite{threshold}. This approach follows the paradigm that one should always try to solve a problem directly and never solve a more general problem as an intermediate step, as advocated by \cite{Vapnik}.  
	
The contribution of this paper is threefold. In Section~\ref{sec:estimation} we define the test statistic $\widehat{ETC}$, an unbiased and consistent estimator of the prediction error of the optimal threshold classifier. The independence of the class-conditional distributions under the null hypothesis is established in Section~\ref{sec:ND}. It is, thus, possible to use the same ND for different types of variables since its distribution depends only on the operating conditions and the sample size. The third and main contribution of this paper is the derivation of the exact ND by means of a fast recursive algorithm, which is introduced in Sections~\ref{sec:ND} and~\ref{sec:algorithm}. A few aspects of its implementation are commented on in Section~\ref{sec:implementation}. Section~\ref{sec:SimulationStudies} sheds light on the efficiency and robustness of the method by means of extensive simulation studies. The final conclusions are drawn in Section  ~\ref{sec:Conclusions}.

\section{Preliminaries}\label{sec:Preliminaries}
	
Consider the random variable $(X,Y)$ defined on $\mathcal{X} \times \{0,1\}$, where $X$ is real valued ($ \mathcal{X} \subseteq \mathbb{R}$) and $Y$ represents the class membership. Let the random variable $X$ conditional on $Y$ be denoted $X^0 := X|Y=0$ and $X^1 := X|Y=1$ and their distributions $F^0$ and $F^1$, respectively. A classifier $ \delta: \mathcal{X} \to \{0,1\} $ is a mapping that attaches a class membership to a realization of $X$. In this paper we will consider a simple type of univariate classifier called threshold classifiers, 

\begin{equation*}\label{eq:etc1}\begin{array}{ll}
\delta_{<t}(x) := \mathbbm{1}_{(-\infty, t)}(x),& t \in \mathbb{R}\\
\delta_{\geq t}(x) := \mathbbm{1}_{[t, \infty)} (x),& t \in \mathbb{R}.\\
\end{array}\end{equation*}
The family of classifiers of this type shall be denoted $\mathcal{F}:= \big \{\delta_{<t}(x), t \in \mathbb{R} \big \} \cup \big \{ \delta_{\geq t}(x) , t \in \mathbb{R} \big \}$. The selection of a classifier $\delta$ from this family shall be based on the expected prediction error, a functional defined by
\begin{equation}\label{eq:defEPE} EPE(\delta) := \mathbb{E}_{(X, Y)} [\textrm{L}(\delta(X), Y)], \end{equation}
where $L$ denotes the loss function
\begin{equation*} \textrm{L}\label{eq:Loss}
(\delta(X),Y):= \left \{ \begin{matrix}
c_1, & \delta(X) = 0 \land Y=1 \\
c_0, & \delta(X) = 1 \land Y=0 \\
0, & \delta(X) = Y \\
\end{matrix} \right \}, \end{equation*}
and $c_0 \geq 0$ and $c_1 \geq 0$ denote the misclassification costs of negative ($Y = 0$) and positive ($Y = 1$) instances, respectively. Let us further introduce the notion of an operating condition (OC), consisting of a triple $\theta=(c_0, c_1, \pi_1)$, where $\pi_1$ denotes the share of positives in the population. The share of negatives $\pi_0$ equals $ 1 - \pi_1$. The OCs parametrize the circumstances under wich a classifier is applied, see \citet{FlachEtAl2012}. Applying Bayes' theorem we can rewrite equation \eqref{eq:defEPE} as a function of these operating conditions:
\begin{equation} \label{eq:EPE} \textrm{ EPE}(\delta) = c_0 \pi_0 \underbrace{P[\delta(X)=1|Y=0]}_{\textrm{false positive rate}} + \: c_1 \pi_1 \underbrace{P[\delta(X)=0|Y=1]}_{\textrm{false negative rate}}.\end{equation}
	
Finally, let us define $ETC$, the expected prediction error of the optimal threshold classifer
\begin{equation}\label{eq:thETCdef} ETC := \min_{\delta \in \mathcal{F} } EPE(\delta). \end{equation}
	
This statistic quantifies the degree to which the classes can be distinguished by setting a threshold on the sample space of $X$. If the sample allows for a perfect separation by means of a threshold, the value will be 0. The maximum value equals $\min \left ( c_0 \pi_0, c_1 \pi_1 \right )$ and is obtained when the conditional distributions $F^0$ and $F^1$ are identical.

\section{Sample Estimate for ETC}\label{sec:estimation}
	
Suppose ${(x_i, y_i), i=1, \ldots, n}$, are i.i.d. realizations of $(X,Y)$. Let $n_1 := \sum_{i=1}^n y_i$ and $n_0 := n - n_1$ denote the number of positives and negatives, respectively. Without loss of generality let us assume that $y_i=0$ for $i = 1,\ldots, n_0$ and $y_i = 1$ for $i = n_0 + 1,\ldots, n$. 
	
To derive a sample estimate of $ETC$ we simply need to substitute the false positive and false negative rates in \eqref{eq:EPE} by their empirical counterparts. For the case of threshold classifiers, these are simply the empirical cumulative distribution functions of the class-conditionals, yielding 
	
\begin{equation}\label{eq:EPEofTC}
\begin{array}{l}
\widehat{EPE} (\delta_{< t}) = c_0 \pi_0 \underbrace{ \frac{1}{n_0} \sum_{i: y_i=0} \mathbbm{1}_{(-\infty, t)} (x_i) }_{\textrm{false positive rate}} + \: c_1 \pi_1 \underbrace{  \frac{1}{n_1} \sum_{i:y_i=1} \mathbbm{1}_{[t, \infty)} (x_i) }_{\textrm{false negative rate}} \\
\widehat{EPE} (\delta_{\geq t}) = c_0 \pi_0 \underbrace{ \frac{1}{n_0} \sum_{i: y_i=0} \mathbbm{1}_{[t, \infty)} (x_i)   }_{\textrm{false positive rate}} + \: c_1 \pi_1 \underbrace{\frac {1}{ n_1 } \sum_{i: y_i=1} \mathbbm{1}_{(-\infty, t)} (x_i) }_{\textrm{false negative rate}}. \end{array}
\end{equation}
	
Note that $\widehat{EPE}(\delta_{<t})$ is constant for all thresholds $ t \in [x_{(j)} ,x_{(j+1)})$, where $x_{(j)}$ denotes the $j$-th order statistic of the sample. Thus, instead of having to search $\mathbb{R}$ it suffices to find the minimum of $ n + 1 $ values. Exploiting this fact leads us to the empirical counterpart of \eqref{eq:thETCdef},
	
\begin{equation}\label{eq:ETCdef}
\widehat{ETC} = \min \left \{ \min_{j=1,\ldots,n} \widehat{EPE}(\delta_{<x_{(j)}}) , \min_{j=1,\ldots,n} \widehat{EPE}(\delta_{\geq x_{(j)}}) \right \}.
\end{equation}

Note, that due to the specific form of a threshold classifier its false negative rate $P[\delta_{<t}(X)=1|Y=0]$ is simply $F^0(t)$, the cumulative CCD evaluated at $t$. By the same argument $P[\delta_{<t}(X)=0|Y=1] = 1 - F^1 (t)$ and, thus, $EPE(\delta_{<t})$ is simply a weighted sum of $F^0(t)$ and $F^1(t)$. Since the empirical cumulative distributions in \eqref{eq:EPEofTC} are unbiased and consistent estimators of $F^0(t)$ and $F^1(t)$, so is $\widehat{ETC}$ an unbiased and consistent estimator of $ETC$.

\section{Derivation of the Null Distribution}\label{sec:ND}
	
We might pose the question of the significance of a certain value of $\widehat{ETC}$. In other words, what is the probability of the observed or a more extreme outcome under the null hypothesis that $X$ exhibits no information about the class membership. The null hypothesis and the alternative hypothesis can be stated $H_0 : F^0 = F^1$ and $H_1 : F^0 \neq F^1$, respectively. For the purpose of filtering, providing a p-value gives additional information, since it allows not only to rank the variables but also to decide how many variables should be selected. 
	
The first result that we would like to establish is that the distribution of $\widehat{ETC}$ under the null hypothesis is independent of the CCDs. This important property allows us to calculate the ND once only for any number of tests since it only depends on the operating conditions and the sample size. Let us formulate this in the following theorem.
	
\begin{thm}\label{thm:ind}
	Consider the i.i.d. samples of the class-conditional random variables $X^0_i \sim F^0, i=1, \ldots, n_0$ and $X^1_i \sim F^1, i=n_0+1, \ldots, n$. Under the null hypothesis $H_0: F^0(x) = F^1(x), \forall x \in \mathcal{X}$ the sampling distribution of $ \widehat{ETC} $ is independent of $F^0$ and $F^1$. 
\end{thm}

\begin{proof}
Let us introduce $r$ which maps a sample $(x_i, y_i)_{i=1,\ldots,n}$ to the permutation of class labels ordered by increasing value of $X$
\begin{equation}\label{eq:suffStat} r: (\mathcal{X} \times \{0,1\})^n \mapsto \mathcal{P}_{n_1, n_0} : r ( (x_i, y_i)_{i = 1, \ldots, n} ) = ( y_{(1)} ,y_{(2)} , \ldots,y_{(n)} ), \end{equation}
where $y_{(i)}$ denotes the class label of the $i$-th order statistic, $y_{(i)}:=y_j: x_j=x_{(i)}$. $\mathcal{P}_{n_1, n_0}$ denotes the space of all possible permutations of $n_1$ positive and $n_0$ negative class labels \[ \mathcal{P}_{n_1, n_0} := \left \{ (y_1, \ldots, y_n): y_i \in \{0,1\} \land \sum_{i=1}^{n} y_i = n_1 \right \}.\] The importance of $r$ lies in the fact that although it reduces the information of the sample it still contains all the information needed to calculate $\widehat{ETC}$. This can be shown by the following factorization $ \widehat{ETC} \equiv \widehat{ETC}_{\mathcal{P}} \circ r $.
		
\[ \begin{xy}
\xymatrix{
(\mathcal{X}, \{0,1\})^n \ar[rr]^{\widehat{ETC}} \ar[rd]_r  &     &  \mathbb{R}^+   \\
&  \mathcal{P}_{n_1, n_0}  \ar[ru]_{\widehat{ETC}_{\mathcal{P}}} &
}
\end{xy} \]
		
To prove the validity of this factorization we can rewrite the first argument of \eqref{eq:ETCdef}
		
\begin{align}\nonumber \min_{i=1, \ldots, n} \widehat{EPE}(\delta_{< x_{(i)}}) & = \nonumber \min_{i=1,\ldots,n} c_0 \pi_0 \frac{1}{n_0} \sum_{j=1}^{n_0} \mathbbm{1}_{(-\infty, x_{(i)})} (x_j) + c_1 \pi_1 \frac{1}{n_1} \sum_{j=n_0+1}^{n} \mathbbm{1}_{[x_{(i)}, \infty)} (x_j)\\
\label{eq:factorize} & = \min_{i=1,\ldots,n} c_0 \pi_0 \frac{1}{n_0} \sum_{j=1}^i y_{(j)} + c_1 \pi_1 \frac{1}{n_1} \sum_{j=i+1}^n (1-y_{(j)}), \end{align}
as a function of $ r ( (x_i, y_i)_{i = 1, \ldots, n} ) $. The same holds true for the second argument in \eqref{eq:ETCdef} and, thus, we have the required factorization of $ \widehat{ETC} $. This is similar to the concept of a sufficient statistic and the argumentation is similar to the characterization of sufficiency of \citet{Neyman}.
		
Under the null hypothesis we are effectively drawing independently $n$ times from the same distribution. The i.i.d. condition implies exchangeability and, thus, every order of positives and negatives is equally likely. Irrespective of the distribution of $(X,Y)$ on the sample space, the induced probability distribution of $r((X_i,Y_i)_{i=1,\ldots,n})$ on $\mathcal{P}_{n_1, n_0}$ is always uniform under $H_0$ and, thus, independent of $F_0$ and $F_1$.
\[ \Prob[r((X_i, Y_i)_{i=1,\ldots,n}) = p] = \frac{1} {\binom{n}{n_0}} \quad \forall p \in \mathcal{P}_{n_1, n_0} \]
As a consequence, the probability distribution on $\mathbb{R}^+$ induced by $\widehat{ETC}_{\mathcal{P}}$ must also be independent of $F_0$ and $F_1$ and is identical to the distribution induced by $\widehat{ETC}$. \hfill $\square$
		
\end{proof}
	
A byproduct of this proof is that it allows us to shift the problem of calculating the ND of $\widehat{ETC}$ from the sample space $ (\mathcal{X} \times \{0,1\})^n $ to $\mathcal{P}_{n_1, n_0}$, where we know that the distribution is uniform. Thus, we can calculate the probability of any event by simply counting the number of favorable permutations and dividing by the number of possible permutations ${\binom{n}{n_0}}$. This insight is the basis for the algorithm introduced in Section~\ref{sec:algorithm}, the second key result of this paper.  
	
Let us define the following partition on $\mathcal{P}_{n_1, n_0}$
\begin{equation}\label{eq:partition} \mathcal{P}_{n_1, n_0} = \bigcup_{ \substack{ 0 \geq fn \geq n_1 \\ 0 \geq fp \geq n_0} } S_{fn, fp}, \textrm{ where } \end{equation} 

\begin{equation}\label{eq:set} S_{fn, fp} := \left \{ p \in \mathcal{P}_{n_1, n_0}: \phi(p)=(fn,fp) \right \}, \end{equation}
and $fn = 0,\ldots,n_1$ and $fp = 0,\ldots,n_0$ denote a given number of false negatives and false positives. The function $\phi: \mathcal{P}_{n_1, n_0} \to \{ (fn, fp): fn \in \{ 0, \ldots, n_1 \} \land fp \in \{ 0, \ldots, n_0 \}  \} $ is required since for many permutations the optimal number of false positives and false negatives of $\widehat{ETC}_{\mathcal{P}}$ can be ambiguous. When this is not the case, $\phi$ maps to $(fn, fp)$ given by the optimal threshold of $\widehat{ETC}_{\mathcal{P}}$, thus $ \widehat{ETC}_{\mathcal{P}}(p) = c_0 \pi_0 fp + c_1 \pi_1 fn, \forall p \in S_{fn, fp}$. For the sake of a well defined partition we need to introduce two conventions. For permutations where more than one position of the threshold yields the same minimal prediction error the threshold is set to the smallest index number. Secondly, for permutations where mapping positives to the left or the right of the threshold yields the same $\widehat{EPE}$, assigning the positives to the left is chosen. The exact expression of $\phi$ can be found in Appendix~\ref{sec:factorize}. With the definitions above we can denote the discrete ND under the assumptions made in Theorem~\ref{thm:ind} as 
	
\begin{equation}
\Prob \left[ \widehat{ETC}=c \right] = \sum_{(fn, fp): c_0\pi_0 fp + c_1 \pi_1 fn = c} \frac{|S_{fn,fp}|} {\binom{n}{n_0}}.
\end{equation}
	
An example of a null distribution for $n_0 = 9$, $n_1 = 9$, $c_0 = 1$, $c_1 = 2$, $\pi_1=0.5$ can be found in Figure~\ref{fig:ND}. We will use this example in the following section to explain key features of the algorithm. 
	
\begin{figure}
\centering
\makebox{\includegraphics[width=1\textwidth]{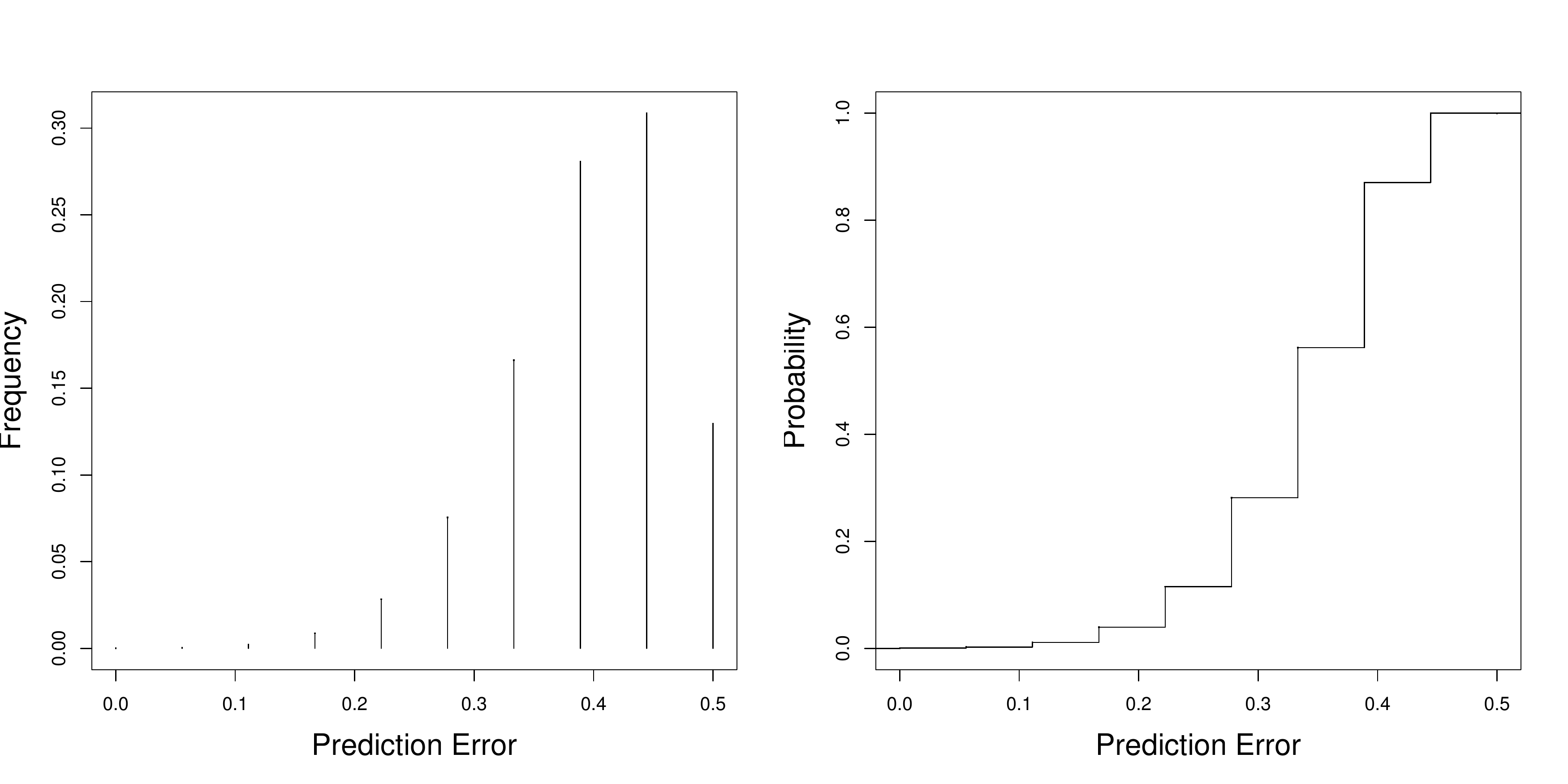}}
\caption{\label{fig:ND}The null distribution (on the left) and the cumulative null distribution (on the right) of $\widehat{ETC}$ for the following operating conditions and sample sizes $n_0 = 9$, $n_1 = 9$, $c_0 = 1$, $c_1 = 2$, $\pi_1=0.5$.} 
\end{figure}

\section{The Algorithm}\label{sec:algorithm}
	
The algorithm described in this section has been devised to count the number of permutations in any set $S_{fn,fp}$. It is based on three main properties, which shall be stated here without a formal proof. Firstly, the sets defined in \eqref{eq:set} can further be split into sets of permutations with the positive domain on the left and on the right side of the threshold,

\begin{equation}\label{eq:subpartition} S_{fn,fp} = S_{fn,fp}^{\leftarrow} \cup S_{fn,fp}^{\rightarrow}.\end{equation}

Secondly, for every permutation $p \in \mathcal{P}_{n_1, n_0} $ the position of the threshold is unambiguous by convention and divides the permutation into a positive and a negative domain with $tp+fp$ and $tn+fn$ instances, respectively. The number of favorable permutations can be counted separately for each domain, since every combination of a favorable permutation of the positive domain $p^+ \in \mathcal{P}_{tp, fp} $ and a favorable permutation of the negative domain $ p^- \in \mathcal{P}_{fn, tn}$ forms a valid permutation of the respective set $p = (p^+, p^-) \in S_{fn,fp}^{\leftarrow}$. Thus, the overall number of permutations is given by the product of the number of favorable permutations of the positive domain and the negative domain. 
	
\begin{align} \label{eq:split} | S_{fn,fp}^{\leftarrow} | & = | S_{fn,fp}^{+, \leftarrow} | \cdot | S_{fn,fp}^{-, \leftarrow} | \\
| S_{fn,fp}^{\rightarrow} | & = | S_{fn,fp}^{+, \rightarrow} | \cdot | S_{fn,fp}^{-, \rightarrow} |
\end{align}
	
Thirdly, for every set of permutations and for both domains, we can construct a starting permutation, which is characterized by the fact that the false negative and false positive instances are as close to the threshold as possible. Any shift of a false instance towards the threshold will render the position of the threshold suboptimal and the resulting permutation will not be element of the respective set. From this initial permutation, the false positive and false negative instances are shifted away from the threshold as long as the obtained permutation remains in the set and the stopping permutation is reached. An example of a starting and stopping permutation for $S_{(1, 2)}^{\leftarrow}$ is illustrated in Figure~\ref{fig:startPerm}.
	
\begin{figure}
\centering
\includegraphics{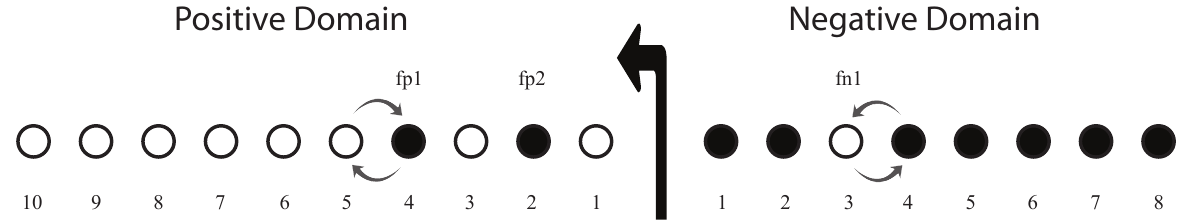}
\includegraphics{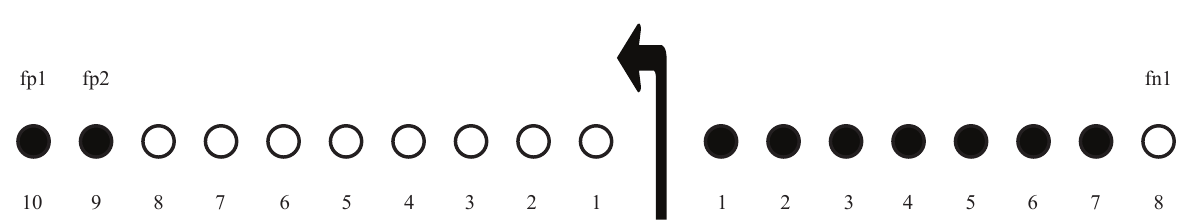}
\caption{\label{fig:startPerm}Illustration of the starting and stopping permutations for both the positive and the negative domain for $S_{(1,2)}$ and for the parameters $n_0 = 9$, $n_1 = 9$, $c_1 = 2$, $c_0=1$, and $\pi_1=0.5$. Black circles indicate negative instances, white circles positive instances. The bold arrow points in direction of the positive domain. The numbers indicate the position indices of the algorithm according to \eqref{eq:count}. Inserting these parameters into \eqref{eq:count_sub1} and \eqref{eq:count_sub2} delivers the starting and stopping indices for the positive domain of $S_{(1,2)}$. The two false positive instances can shift from position 4 to 10 and 2 to 9, respectively. According to \eqref{eq:count_sub1} and \eqref{eq:count_sub2} this yields $\sum_{i_1 = 4}^{10} \sum_{i_2 = 2}^{ \min \{ 9, i_1 - 1\} } 1 = 2+3+4+5+6+7+8 = 35$. As stated in \eqref{eq:count2_sub1} and \eqref{eq:count2_sub2} the false negative instance can shift from position 3 to position 8 in the negative domain, yielding $\sum_3^8 1 = 6$ permutations. According to \eqref{eq:split} the overall number of permutations is, thus, $6 \cdot 35 = 210$.}
\end{figure}

In order to properly describe this algorithm we need to introduce two conventions. The first convention concerns the indexation of the positions of a permutation. As illustrated in Figure~\ref{fig:startPerm}, the indices are in increasing order starting at the threshold. The second convention concerns the indexation of the false positives or false negatives. The first false instance will always have the highest position number, the second false instance the second highest and so on. These two conventions allow us to express the number of favorable permutations of the positive domain left of the threshold by the following formula
	
\begin{align}\label{eq:count} 
| S_{fn,fp}^{+, \leftarrow} | & = \sum_{i_{1} = start_{1} }^{stop_{1}} \sum_{i_{2} = start_{2}}^{\min \left \{ stop_2, i_1 - 1 \right \} } \cdots \sum_{i_{fp}=start_{fp}}^{ \min \left \{ stop_{fp}, i_{fp-1} - 1 \right \} } 1,  \textrm{ where }\\
\label{eq:count_sub1} start_k & = \min \left \{ tp + fp - (k-1) , l_1 + (fp - k) + 1\right \}, \\ 
l_1 & = \min \left \{ n \in \mathbb{N} : n c_1\pi_1 > ( fp - k +1 ) c_0  \pi_0\right \}, \nonumber \\
\label{eq:count_sub2} stop_k & = \min \{ tp + fp - (k-1) , start_{k} + l_2 \}, \\
l_2 & =  \max \left \{ n \in \mathbb{N} : fp~c_0 \pi_0 + fn~c_1 \pi_1 \leq ( tn + fp - k) c_0 \pi_0 + \left ( tp - l_1 - n \right ) c_1 \pi_1  \right \}. \nonumber 
\end{align}
	
The parameter $start_k$ in \eqref{eq:count_sub1} denotes the starting position of the $k$-th false positive. For the threshold to be optimal there must be at least $\min \{n \in \mathbb{N} : n c_1 \pi_1 > k c_0 \pi_0 \}$ positives to the right plus another $fp - k$ negatives. The index of the $k$-th false positive is, however, limited to $tp - fp - (k-1)$, which is the case when no positive instances are to the left. 
	
The parameter $stop_k$ in \eqref{eq:count_sub2} denotes the final stopping position of the $k$-th false positive. It equals the starting index plus $l_2$, which denotes the maximal number of shifts the positive instance can undertake before the threshold becomes suboptimal. In this case the positive domain would switch to the other side and its position would switch to $start + l_2 + 1$. However, the stopping index can never exceed $tp + fp - (k-1)$ because at this position there are no more positives to the left of the $k$-th false positive. 

Equivalent recursive expressions for $| S_{fn,fp}^{+, \rightarrow} |$, $| S_{fn,fp}^{-, \leftarrow}|$, and $| S_{fn,fp}^{-, \rightarrow}|$ can be found in \eqref{eq:count2}, \eqref{eq:count3}, and \eqref{eq:count4} in Appendix~\ref{sec:schema}. 
	
\section{Implementation}\label{sec:implementation}

Since the starting and stopping indices of the sum of index $i_k$ depend on the current value of the index $i_{k-1}$ and the number of sums depends on the number $fp$ and $fn$, equation \eqref{eq:count} is implemented as a recursive scheme. The initial recursion level equals the number of false positives or false negatives. The recursive function calls itself, recalculating the start and stop indices for the lower level until the base case is reached. In this case, the algorithm simply returns $stop_{i_{fp}} - start_{i_{fp}}$.

The described method has been implemented for the R software platform \citep{R} and is part of the \texttt{CVOC} package \citep{CVOC} that has been published unter the GNU general public license. The challenges of implementing it were twofold. First, since it is recursive in nature, the number of function calls grows exponentially as the sample size increases and so does the required computation time, which is $O(2^n)$, where n denotes the sample size. Especially at lower recursion levels the function will be called repeatedly with the exact same arguments. This can be avoided by means of memoization or dynamic programming, programming techniques which store the input arguments and the output in a suitable data object and use the cached results instead of calling the function again. Furthermore, the recursive functions were implemented in C++ using the \texttt{Rcpp} package, see \citet{Rcpp}. The modified algorithm is of the order $ O(1.05^n)$. The performance gains resulting from these modifications are illustrated in Figure~\ref{fig:benchmark}.

\begin{figure}
	\centering
	\includegraphics[width=0.6\textwidth]{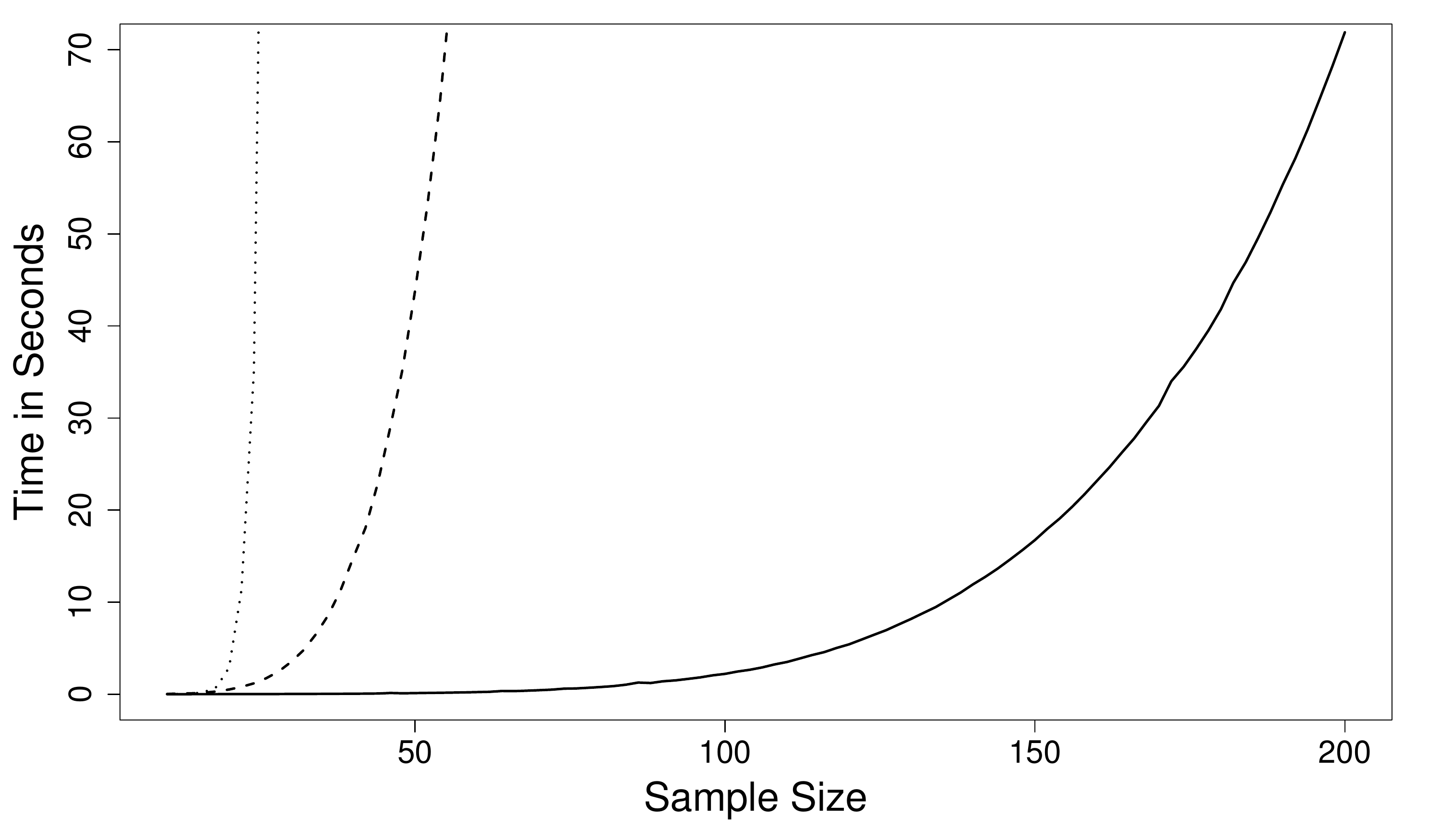}
	\caption{\label{fig:benchmark} This figure illustrates the time necessary for computing the ND on a standard PC with 4 kernels in seconds as a function of the sample size. The dotted line represents the standard algorithm implemented in R. The dashed line introduces memoization to the algorithm within the R platform. The solid line, which clearly, outperforms the others, represents the algorithm implemented in C++ including memoization.}
\end{figure}

Secondly, as the sample size increases, the number of possible permutations quickly exceeds the maximum integer value. This happens when the sample contains more than 30 instances. Thus, high precision numbers are needed to run the algorithm and save its results. The GNU multiple precision arithmetic library GMP \citep{GMP} and the GNU MPFR \citep{MPFR1} library were used. 
	
Another challenge concerns the existence of ties in the data. Theoretically, if the distribution of $X$ is continuous, there will almost surely be no tied order statistics. However, in real data applications, there might be instances with the same value and different class labels. The implemented algorithm proceeds by calculation both test statistics and returning the more conservative value.

\section{Simulation Studies}\label{sec:SimulationStudies}
	
This section intends to shed light on the capacity of ETC and comparable classifiers to select signal variables out of a large number of noise variables. A signal variable is one that exhibits differing CCDs. In the following simulation studies, 1000 signal variables and 99000 noise variables are repeatedly generated and the percentage of signal variables among the first 1000 ranked variables - this percentage will be referred to as the filtering performance (FP) - was observed. 
	
The methods under scrutiny are ETC, the linear discriminant analysis (LDA) and the quadratic discriminant analysis (QDA). Univariate LDA is a parametric threshold classifier which assumes two Gaussian CCDs with equal variances, estimates the parameters and then derives the optimal threshold given these parameters and the operating conditions. If the assumptions are met, the LDA is the Bayes classifier and will constitute the most efficient variable filter. In simulation study A, we will, thus, analyse the loss in efficiency of ETC and QDA as compared to the benchmark. If the assumption of equal variances is relaxed, as done in simulation study B, the univariate QDA becomes the Bayes classifier, which is, however, no threshold classifier. Under these assumptions the optimal positive and negative domain is an interval and its respective complement. In this setting we can analyse to what extent the variances of the CCDS need to differ to manifest in QDA becoming the most efficient filter. While simulation study C considers data that are contaminated with varying degrees of outliers simulation study D considers data that stem from non-symmetric or skewed distributions. Thus, the overall purpose of these sections is to compare the robustness of ETC to its parametric counterpart LDA and QDA. 
	
\subsection{Simulation Study A}
	
In this experiment, characterized in \eqref{eq:SimA}, the signal variables with a sample size of $n \in [10, 200]$, where $n_1 = n_0$, are drawn from two Gaussian distributions $ N(\mu_1, 1)$ and $ N(\mu_0, 1)$, which differ only with respect to their central location, $\Delta \mu = \mu_1 - \mu_0 \in [0, 2.5]$.
	
\begin{equation}\label{eq:SimA} \begin{array}{rlll}
X_{ij} & \sim & \mathcal{N}(\mu_1, 1) \qquad & i = 1, \ldots, n_1; j = 1, \ldots, 1000 \\
X_{ij} & \sim & \mathcal{N}(\mu_0, 1) \qquad & i = n_1 + 1, \ldots, n; j = 1, \ldots, 1000 \\
X_{ij} & \sim & \mathcal{N}(0,1) \qquad & i = 1, \ldots, n; j = 1001, \ldots, 100000
\end{array} \end{equation}
	
Under these assumptions, the LDA is the Bayes classifier and will serve as the benchmark to evaluate the loss in filtering performance of ETC and QDA. The leftmost image in Figure~\ref{fig:SSA} shows the filtering performance of LDA. The two images to the right illustrate the loss in FP of ETC and QDA, respectively. The two latter show a very similar pattern. The loss in performance increases steadily as $\Delta \mu$ increases to about 40 \% and then converges to 0 with increasing $\Delta \mu$ depending on the sample size.

\begin{figure}
\centering
\includegraphics[width=1\textwidth]{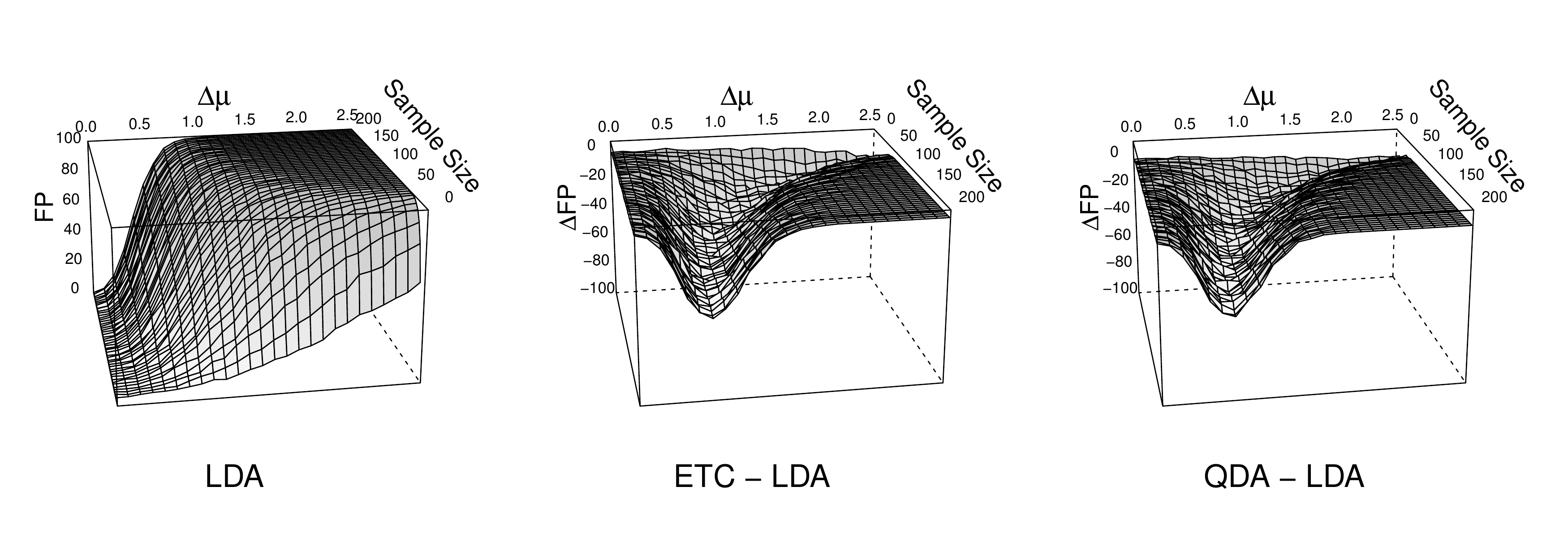}
\caption{\label{fig:SSA}Simulation study A analyses the filtering performance of the three classifiers under scrutiny for two Gaussian class conditional distributions with equal variances. The leftmost image illustrates the FP of LDA, the benchmark for this simulation setting. The two other images depict the difference in the filtering performance ($\Delta$FP) of ETC and QDA to the benchmark, respectively.}
\end{figure}

\subsection{Simulation Study B}
	
In this experiment, characterized in \eqref{eq:SimB}, the sample size is fixed to $n=100$ and the assumption of equal variances was relaxed, $\sigma_1 \in 2^{[-3, 3]}$ and $\sigma_0 = 1/\sigma_1$. As in simulation study A the difference in the central location differ, $\Delta \mu = \mu_1 - \mu_0 \in [0, 2.5]$. The results can be seen in Figure~\ref{fig:SSB}.
	
\begin{equation}\label{eq:SimB} \begin{array}{rll}
X_{ij} \sim & \mathcal{N}(\mu_1, \sigma_1^2) \qquad & i = 1, \ldots, 50; j = 1, \ldots, 1000 \\
X_{ij} \sim & \mathcal{N}(\mu_0, \sigma_0^2) \qquad & i = 51, \ldots, 100; j = 1, \ldots, 1000 \\
X_{ij} \sim & \mathcal{N}(0,1) \qquad & i = 1, \ldots, 100; j = 1001, \ldots, 100000 \\
\end{array} \end{equation}
	
Simulation study B indicates that there are systematic differences in the performance of the classifiers when the assumption of equal variances is not met. Intuitively, the stronger the CCDs differ from another, the easier it should be to predict the class based on the observation of $X$. This characteristic can be observed for QDA, which is the Bayes classifier under these assumptions. The greater the difference between $\sigma_1$ and $\sigma_0$ the better the filtering performance of QDA. Under these assumptions the Bayes classifier is not a member of the family of threshold classifiers, since the optimal positive domain is an interval or its complement. It is, thus, interesting to observe at what ratio of $\frac{\sigma_1}{\sigma_0}$ the false assumption of LDA and ETC manifest in the filtering performance. While the filtering performance of LDA clearly deteriorates when this central assumption is not met, ETC can benefit from differences in the variances of the CCDs if the positive class exhibits the relatively lower variance and $c_1 > c_0$. This becomes clear, if one considers the extreme case of $\sigma_1 = 0$, where the positive CCD degenerates to a single point. Even if $\Delta \mu = 0 $, one can obtain a false negative rate of 0 and a false positive rate of 0.5 by setting a threshold just below this point.

\begin{figure}
\centering
\includegraphics[width=1\textwidth]{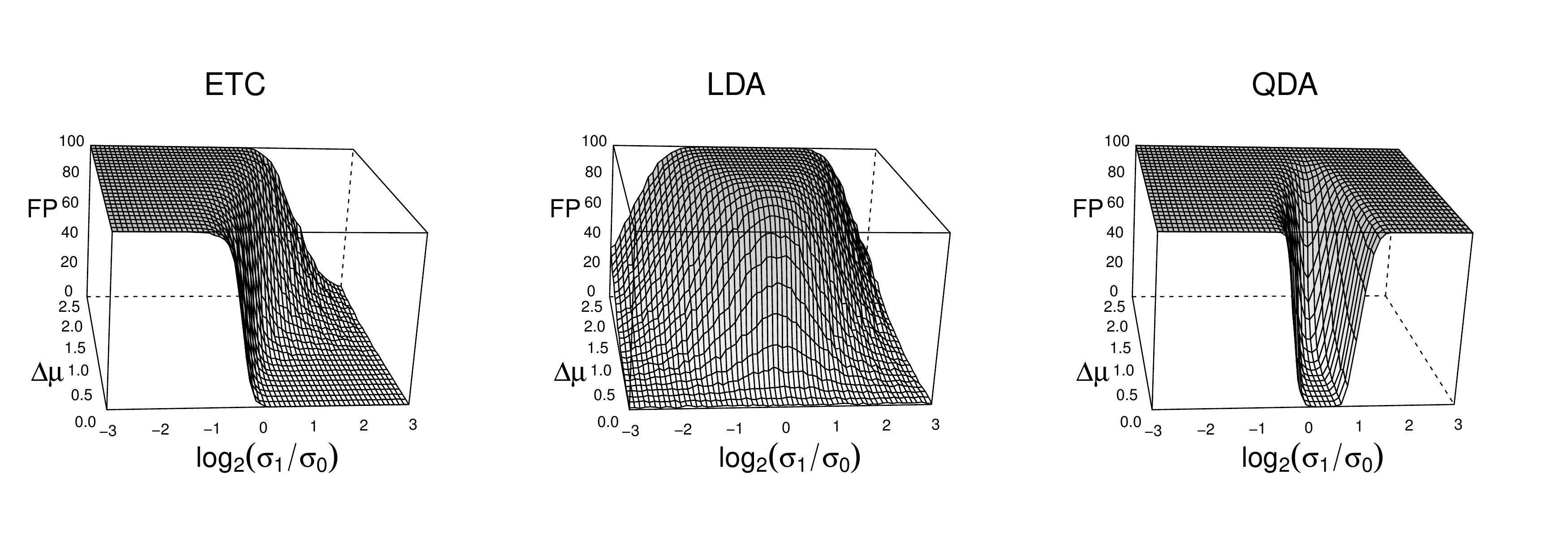}
\caption{\label{fig:SSB} Simulation study B analyses the filtering performance of the three classifiers for two Gaussian class conditional distributions for varying means and variances.} 
\end{figure}

\subsection{Simulation Study C}
	
This experiment, characterized in \eqref{eq:SimC}, considers data that are contaminated with a varying proportion of outliers, $\phi \in [0, 0.3]$. Furthermore, we shall vary the difference in the central location, $\Delta \mu = \mu_1 - \mu_0 \in [0, 2.5]$. Again the number of samples are fixed at $n = 100$, with an equal number of positives and negatives. In order to blend out the differences in the filtering performance caused by other effects we shall compare the FP with the FP of the data without outliers. Thus, let us define $\Delta$FP$(\phi):=$FP$(\phi) - $FP$(0)$, where FP$(\phi)$ denotes the filtering performance with a proportion of $\phi$ outlying values, ceteris paribus. The results are depicted in Figure~\ref{fig:SSC}.  
	
\begin{equation} \label{eq:SimC} \begin{array}{rlll}
X_{ij} & \sim & \mathcal{N}(\mu_1, 1) & i = 1, \ldots, 50 - \left \lfloor \phi \cdot 50 \right \rfloor; j = 1, \ldots, 1000 \\
X_{ij} & \sim & \mathcal{N}(\mu_1, 5) & i = \left \lfloor \phi \cdot 50 \right \rfloor + 1, \ldots, 50; j = 1, \ldots, 1000 \\
X_{ij} & \sim & \mathcal{N}(\mu_0, 1) & i = 51, \ldots, 100 - \left \lfloor \phi \cdot 50 \right \rfloor; j = 1, \ldots, 1000 \\
X_{ij} & \sim & \mathcal{N}(\mu_0, 5) & i = 100 - \left \lfloor \phi \cdot 50 \right \rfloor + 1, \ldots, 100; j = 1, \ldots, 1000 \\
X_{ij} & \sim & \mathcal{N}(0,1) & i = 1, \ldots, 100 - \left \lfloor \phi \cdot 100 \right \rfloor; j = 1001, \ldots, 100000 \\
X_{ij} & \sim & \mathcal{N}(0,5) & i = 100 - \left \lfloor \phi \cdot 100 \right \rfloor + 1, \ldots, 100; j = 1001, \ldots, 100000 \\
\end{array} \end{equation}
	
Simulation Setting C sheds light on the robustness of the classifiers in the presence of outliers. The filtering performance deteriorates in all three methods, however ETC proves to be the most robust. 
	
\begin{figure}
\centering
\includegraphics[width=1\textwidth]{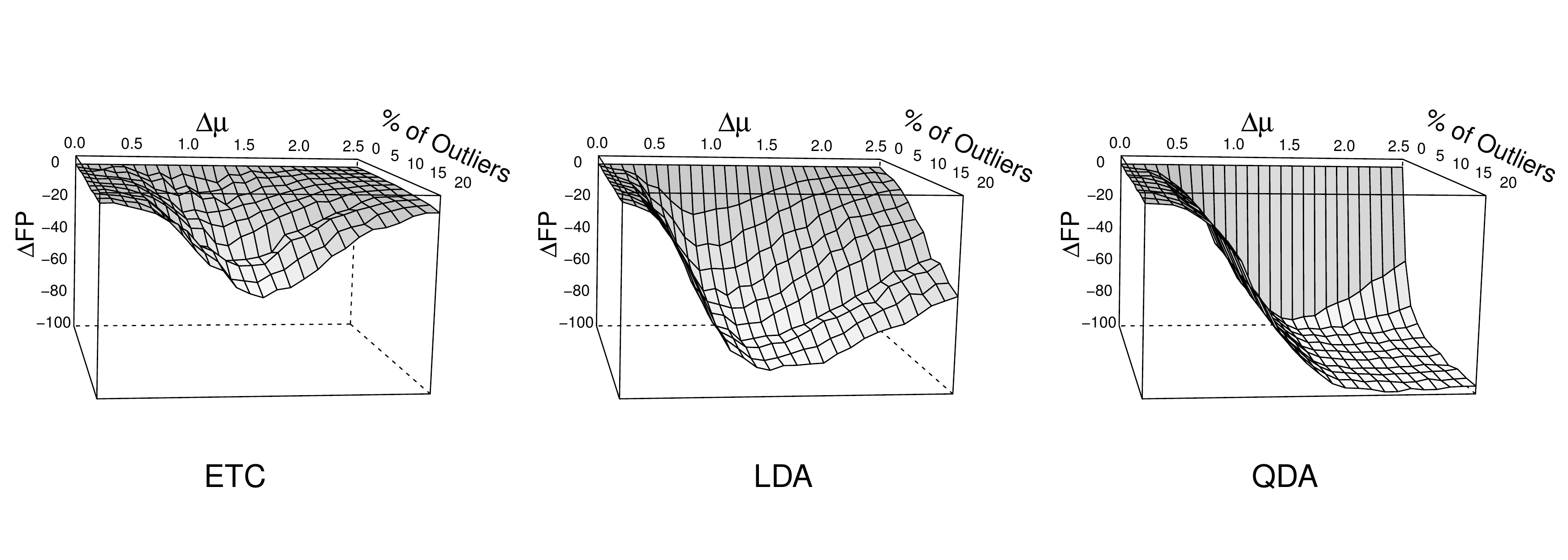}
\caption{\label{fig:SSC}Simulation study C analyses the filtering performance of ETC, LDA, and QDA for two Gaussian class conditional distributions that are contaminated with a varying degree of outliers. $\Delta$FP denotes the difference in performance as compared to the simulation study without outliers all other parameters being equal.} 
\end{figure}

\subsection{Simulation Study D}
	
In this simulation study, characterized in \eqref{eq:Sim3}, we would like to study the effect of skewness on the filtering performance. The random variables are, thus, drawn from a log-normal distribution with the parameters $\mu=0$ and $ \sigma \in \sqrt{2^{[-3, 3]}}$. In order to introduce a difference in the central location, the parameter $ \Delta \in [0, 2.5] $ controls the shift in the central location of the distribution. Again, the number of samples is fixed at $n=100$, with $n_1=n_0=n/2$. The results are depicted in Figure~\ref{fig:SSD}.
	
\begin{equation}\label{eq:Sim3} \begin{array}{rlll}
\ln (X_{ij}) & \sim & \mathcal{N}(0, \sigma^2) \qquad & i = 1, \ldots, n_1; j = 1, \ldots, 1000 \\
\ln (X_{ij}) & \sim & \mathcal{N}(0, \sigma^2) - \Delta \qquad & i = n_1 + 1, \ldots, n; j = 1, \ldots, 1000 \\
\ln (X_{ij}) & \sim & \mathcal{N}(0,\sigma^2) \qquad & i = 1, \ldots, n; j = 1001, \ldots, 100000 \\
\end{array} \end{equation}
	
Simulation Setting D sheds light on the filtering performance of the three classifiers under scrutiny when the CCDs are skewed. While ETC is completely resilient against deviations from symmetry, LDA and QDA clearly show a deterioration in the filtering performance.
	
\begin{figure}
\centering
\includegraphics[width=1\textwidth]{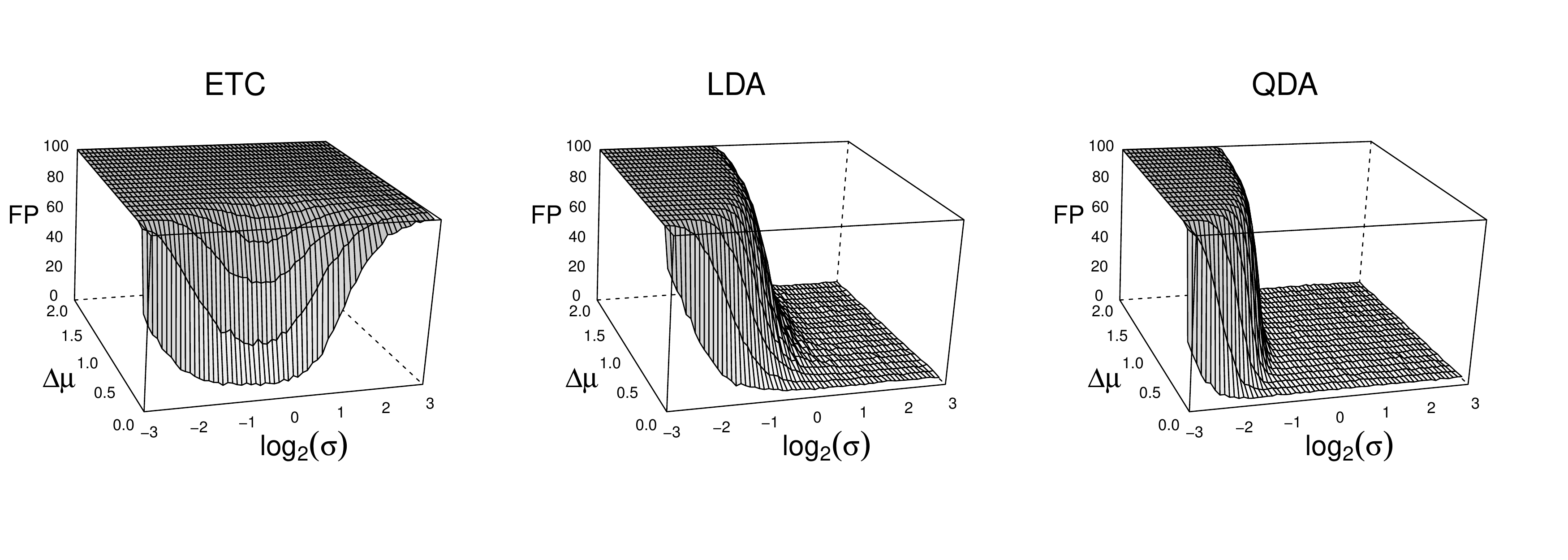}
\caption{\label{fig:SSD}Simulation study D analyses the filtering performance of the three classifiers for two lognormal class conditional distributions with varying degrees of skewness.} 
\end{figure}

\section{Conlusions}\label{sec:Conclusions}
This paper introduces a novel test for the separability of two classes by means of a simple threshold. Its test statistic is the prediction error of the nonparametric threshold classifier. We have derived an unbiased and consistent estimator for this test statistic and devised a fast recursive algorithm with which one can calculate its exact finite sample distribution under the null hypothesis of equal class conditional distributions. 

This test constitutes an exact nonparametric method which, thus, exhibits advantageous properties for the purpose of filter-type variable selection. First, we have proven that its ND is independent of the class conditionals and, thus, this test yields exact results for data following different kinds of distributions. Furthermore, simulation studies have shown that it exhibits only a small loss in efficiency compared to the Bayes classifier and that this loss is mostly outweighed by its robustness towards outlying values and skewed distributions. ETC will, thus, select important variables which may otherwise not be found by other filtering methods.
	
The derivation of its ND required the development of a novel type of algorithm. This approach is, however, not limited to the threshold classifier and can be extended to other methods as long as \eqref{eq:suffStat} preserves all the information and the distribution on $\mathcal{P}_{n_1, n_0}$ is uniform. The most obvious extension is to the family of interval-type classifiers $ \mathcal{F} = \big \{\delta_{(t_1, t_2]}(x) = \mathbbm{1}_{(t_1, t_2]}(x), t_1, t_2 \in \mathbb{R} \big \} \cup \big \{ \delta_{\mathbb{R} \setminus (t_1, t_2]}(x) = \mathbbm{1}_{\mathbb{R} \setminus (t_1, t_2] }, t_1, t_2 \in \mathbb{R} \big \}$. The Bayes classifier for two Gaussian class-conditionals with different variances is of this type. Furthermore, the idea seems to allow an application to multivariate nonparametric classifiers where the p-value could serve as an information criteria for model selection. This will be the subject of investigation in the foreseeable future.
	
\appendix

\section{Defining the Partition of $\mathcal{P}_{n_1, n_0}$}\label{sec:factorize}

\begin{equation*} \phi(p) := \left \{ \begin{matrix}
(\sum\limits_{j=i}^{n} p_j, \sum\limits_{j=1}^{i-1} 1-p_j), i = \min \argmin\limits_{j=1, \ldots, n} E_{<j}(p) & \textrm{ if } \min\limits_{j=1, \ldots, n} E_{<j}(p) \leq \min\limits_{j=1, \ldots, n} E_{\geq j}(p) \\
(\sum\limits_{j=1}^{i-1} p_j, \sum\limits_{j=i}^{n} 1-p_j), i = \min \argmin\limits_{j=1, \ldots, n} E_{\geq j}(p) & \textrm{ else }
\end{matrix} \right.
\end{equation*}

\[ E_{<i}(p) := c_0 \pi_0 \frac{1}{n_0} \sum_{j=1}^{i-1} 1 - p_j + c_1 \pi_1 \frac{1}{n_1} \sum_{j=i}^{n} p_j \]

\[ E_{\geq i}(p) := c_0 \pi_0 \frac{1}{n_0} \sum_{j=i}^{n} 1 - p_j + c_1 \pi_1 \frac{1}{n_1} \sum_{j=1}^{i-1} p_j \]

\section{The Recursive Counting Schemes}\label{sec:schema}
	
\begin{align}\label{eq:count2}
|S_{fn,fp}^{-, \leftarrow} | & = \sum_{j_{1}=start_{1}}^{ stop_1 } \sum_{j_{2} = start_{2}}^{ \min \{ stop_2, j_{1} - 1\} } \cdots \sum_{j_{fn}=start_{fn}}^{ \min \{ stop_{fn}, j_{fn - 1} - 1 \}  } 1,  \textrm{ where }\\
\label{eq:count2_sub1} start_k & = \min \left \{ tn + fn - (k-1) , l_1 + (fn - k) + 1\right \}, \\
l_1 & = \left \{ n \in \mathbb{N}: n c_0 \pi_0 \geq ( fn - k +1 ) c_1 \pi_1 \right \}, \nonumber \\
\label{eq:count2_sub2} stop_k & = \min \{ tn + fn - (k-1) , start_{k} + l_2 \},  \\
l_2 & = \max \{ n \in \mathbb{N} : fp~c_0 \pi_0 + fn~c_1 \pi_1 < \left ( tn - l_1 - n \right ) c_0 \pi_0 + ( tp + fn - k) c_1 \pi_1 \}. \nonumber
\end{align}

\begin{align}\label{eq:count3} 
| S_{fn,fp}^{+, \rightarrow} | & = \sum_{i_{1} = start_{1} }^{stop_{1}} \sum_{i_{2} = start_{2}}^{\min \{ stop_2, i_1 - 1 \} } \cdots \sum_{i_{fp}=start_{fp}}^{ \min \{ stop_{fp}, i_{fp-1} - 1 \} } 1,  \textrm{ where }\\
start_k & = \min \left \{ tp + fp - (k-1) , l_1 + (fp - k) + 1\right \}, \nonumber \\ 
l_1 & = \min \left \{n \in \mathbb{N}: n c_1 \pi_1 \geq (fp - k +1 ) c_0 \pi_0 \right \}, \nonumber \\
stop_k & = \min \left \{ tp + fp - (k-1) , start_{k} + l_2 \right \},  \nonumber \\
l_2 & = \max \left \{n \in \mathbb{N} : fp~c_0 \pi_0 + fn~c_1 \pi_1 < \left ( tp - l_1 - n \right ) c_1 \pi_1 + ( tn + fp - k) c_0 \pi_0 \right \}. \nonumber
\end{align}

\begin{align}\label{eq:count4}
|S_{fn,fp}^{-, \rightarrow} | & = \sum_{j_{1}=start_{1}}^{ stop_1 } \sum_{j_{2} = start_{2}}^{ \min \{ stop_2, j_{1} - 1\} } \cdots \sum_{j_{fn}=start_{fn}}^{ \min \{ stop_{fn}, j_{fn - 1} - 1 \}  } 1, \textrm{ where }\\
start_k & = \min \left \{ tn + fn - (k-1) , l_1 + (fn - k) + 1\right \}, \nonumber\\ 
l_1 & =  \min \left \{n \in \mathbb{N}: n c_0 \pi_0 > ( fn - k +1 ) c_1 \pi_1 \right \}, \nonumber \\
stop_k & = \min \{ tn + fn - (k-1) , start_{k} + l_2 \},  \nonumber\\
l_2 & = \max \{n \in \mathbb{N} : fp~c_0 \pi_0 + fn~c_1 \pi_1 \leq \left ( tn - l_1 - n \right ) c_0 \pi_0 + ( tp + fn - k) c_1 \pi_1 \}. \nonumber
\end{align}

\bibliographystyle{plainnat}
\bibliography{refs}

\end{document}